\documentclass{llncs}

\usepackage[pdftex]{graphicx}
\usepackage{wrapfig}

\usepackage{algorithm}
\usepackage{algpseudocode}

\usepackage[utf8]{inputenc} 
\usepackage[T1]{fontenc}    
\usepackage{hyperref}       
\usepackage{url}            
\usepackage{booktabs}       
\usepackage{amsfonts}       
\usepackage{nicefrac}       
\usepackage{microtype}      
\usepackage{xcolor}         
\usepackage{soul}

\usepackage{amssymb}
\usepackage{mathtools}
\usepackage[capitalize,noabbrev]{cleveref}
\usepackage{bbm}
\usepackage{tabularx}
\usepackage{listings}

\usepackage{enumitem}

\title{Differentially Private Empirical Cumulative Distribution Functions}

\author{
  Antoine Barczewski,
  Amal Mawass
  \and
  Jan Ramon
}

\institute{
MAGNET INRIA Lille, France}

\begin{document}

\maketitle

\begin{abstract}
  
In order to both learn and protect sensitive training data, there has been a growing interest in privacy preserving machine learning methods.  Differential privacy has emerged as an important measure of privacy.  We are interested in the federated setting where a group of parties each have one or more training instances and want to learn collaboratively without revealing their data.

In this paper, we propose strategies to compute differentially private empirical distribution functions.  While revealing complete functions is more expensive from the point of view of privacy budget, it may also provide richer and more valuable information to the learner.  We prove privacy guarantees and discuss the computational cost, both for a generic strategy fitting any security model and a special-purpose strategy based on secret sharing.  We survey a number of applications and present experiments.

\end{abstract}

\newcommand{\janfoot}[1]{}
\newcommand{\amalfoot}[1]{}
\newcommand{\antoinefoot}[1]{}

\section{Introduction}

Over the last years, there has been an increasing interest in privacy-preserving machine learning, i.e., learning while protecting the sensitive underlying training data.
Differential privacy \cite{Dwork2014a} has become the gold standard to measure the privacy of an algorithm.  For a wide range of machine learning strategies, versions have been proposed which output models with an appropriate amount of noise to satisfy differential privacy.   A setting of particular interest is the federated learning setting where there is a large number of parties which each own one or more training instances and which want to jointly compute a statistic or model without revealing their own data.

In this paper we are interested in empirical cumulative distribution functions (ECDFs).  Consider a setting where every party in a group has one (or more) values.  Then, the empirical distribution function returns for every threshold the number of values which are smaller than that threshold.  Cumulative distribution functions play an important role in statistics and machine learning, e.g., in tail bounds and in statistics based on rankings. An ECDF which is often used for machine learning validation is the receiver-operator characteristic (ROC) curve.  Often an ECDF is an intermediate result from which a single number is computed, e.g., the area under the ROC curve is an important metric capturing the behavior of a classifier in a single number.  Nevertheless, knowledge of the complete ECDF is in many cases an asset of its own value, e.g., for a ROC curve depending on the cost structure of the problem at hand the beginning or end of the curve may be the most interesting.

We propose differentially private ECDFs.  For a fixed privacy level, publishing a complete ECDF requires more noise than publishing a single aggregated value, but is also much more informative.  Moreover, it turns out that the amount of noise to be added only needs to be logarithmic in the required precision, which is better than pointwise adding independent noise to ECDF evaluations.
Our work here is inspired by but also improves on earlier work on continual observation of statistics in data streams, which was initiated using $\epsilon$-differential privacy by \cite{dwork2010differential} and later extended towards among others indefinite time periods and other privacy notions such as Renyi differential privacy \cite{Chan2011,DBLP:journals/corr/abs-2103-16787}.  Our work is to some extent orthogonal to these extensions: for simplicity of explanation we will adopt the classic $\epsilon$-differential privacy, but our work can be combined with several of these other ideas.

We will also investigate federated algorithms to compute differentially private ECDFs.  In particular, we are interested in both evaluating ECDFs and evaluating inverse ECDFs.  We explore two avenues.  First, we investigate an approach starting from a secure aggregation protocol as a building block.  The advantage of such approach is that it inherits the security guarantees of the aggregation protocol.  One can plug in an aggregator which is cheaper but assumes parties are honest-but-curious, or one can plug in an aggregator which is more expensive but robust against malicious behavior.  The result then is a federated algorithm to compute an ECDF having the same security guarantees.  We hope to contribute in this way to an evolution to a more modular organization of building blocks where there is no need for a different algorithm for every different problem and security setting.
Second, we will investigate strategies to compute ECDF evaluations more efficiently compared to pointwise federated evaluations, at the cost of making some assumptions on the aggregation protocol used.


In summary, our contributions are:
\begin{itemize}
\item We propose a strategy for generating differentially privately a complete ECDF.  We prove that the resulting curve is $\epsilon$-differentially private.
\item In the process of doing so, we make a minor constant-factor improvement over differential privacy guarantees for continually observed statistics as shown in \cite{DBLP:journals/corr/abs-2103-16787}.
\item As a differentially private version of a non-decreasing function is not necessarily non-decreasing itself, we propose a strategy to smooth differentially private functions.  This makes the function non-decreasing again, and in some cases may even reduce the error induced by the differential privacy noise.
\item We discuss strategies for efficient implementations.  In particular, we propose both a generic algorithm which can start from any secure aggregation operator and inherit its security/privacy model and a specific strategy strategy based on secret sharing which has asymptotically better complexity.
\item We present in more detail two applications.  First, we discuss how to make $\epsilon$-differentially private ROC curves using our technique.  Second, we present a differentially private Hosmer-Lemeshow statistic.
\item We illustrate our techniques with a number of experiments.
\end{itemize}

The remainder of this paper is organized as follows.
In Section \ref{sec:prelim} we introduce basic concepts and notations.
Next, in Section \ref{sec:method} we present differentially private ECDFs and prove the corresponding privacy guarantees and in Section \ref{sec:algo} we discuss federated algorithms to compute differentially private ECDFs.
In Section \ref{sec:exp} we present experiments illustrating our approach.
Finally, in Section \ref{sec:concl} we conclude and outline future work.

\section{Preliminaries}
\label{sec:prelim}

\newcommand{\indicVal}[1]{\mathbbm{I}\left[{#1}\right]}
\newcommand{\indicSet}[1]{\mathbbm{1}_{{#1}}}
We denote by $[m,n]=\{z\in\mathcal{Z}\mid m\le z \le n\}$ the set of integers between and including $m$ and $n$, and by $[n]=[1,n]$ the set of the first $n$ positive integers.
For a boolean expresseion $b$, we denote by $\indicVal{b}$ its truth value, i.e., $\indicVal{true}=1$ and $\indicVal{false}=0$.  For a set $X$, we denote its indicator function by $\indicSet{X}$, i.e., $\indicSet{X}(x) = \indicVal{x\in X}$.

We consider datasets $X=\{x_i\}_{i=1}^n\in \mathcal{X}^n$ where $\mathcal{X}$ is a space of instances and $n\in\mathbb{N}$ is a positive integer.
We assume there are $n$ parties $P_i$ with $i\in[n]$, each owning one instance $x_i$.  Our results will generalize easily to the case where every party may own multiple instances.  We assume the instances $x_i$ include sensitive data and the parties don't want to reveal them.  Still, the parties want to collaborate to compute statistics of common interest.

\begin{definition}[differential privacy]
  Let $\epsilon>0$.  Let $\mathcal{A}$ be an algorithm taking as input datasets from $\mathcal{X}^*$.  Two datasets $X^{(1)},X^{(2)}\in \mathcal{X}^*$ are adjacent if they differ in only one element.   The algorithm $\mathcal{A}$ is $\epsilon$-differentially private ($\epsilon$-DP) if for every pair of adjacent datasets $X^{(1)}$ and $X^{(2)}$, and every subset $S$ of possible outputs of $\mathcal{A}$, $P(\mathcal{A}(X^{(1)})\subseteq S) \le e^\epsilon P(\mathcal{A}(X^{(2)})\subseteq S)$.
\end{definition}

\newcommand{\ffeat}{\phi}

\begin{definition}[U-statistic]
  Let $m \ge 1$ be a positive integer and let $\ffeat :\mathcal{X}^{m} \to \mathbb{R}$ be a symmetric function.  The $U$-statistic with kernel $\ffeat$ is the function mapping samples $\{x_i\}_{i=1}^n\in\mathcal{X}^n$ on
  \begin{equation}
    U_\ffeat(X)={\binom {n}{m}}^{-1}\sum _{1\le i_{1}<\dots <i_{r}\le n}\ffeat (x_{i_{1}},\dots ,x_{i_{m}}).
  \end{equation}
  We say $U_\ffeat$ is a $U$-statistic of order $m$.
\end{definition}
Of particular interest are $U$-statistics of order $1$, which are averages of the form $U_\ffeat(X)= (1/n) \sum_{i=1}^n \ffeat(x_i)$.

\newcommand{\rvseq}{\eta}
\newcommand{\rvidxs}{I}
\newcommand{\rvidxu}{\mathcal{I}}
\newcommand{\rvidxL}[1]{\mathcal{I}[{#1}]}
\newcommand{\rvidxLI}[2]{\mathcal{I}[{#1},{#2}]}
Let $\rvseq=(\rvseq_i)_{i\in\rvidxu}$ be a vector of random variables for some set $\rvidxu$ of indices.  Then, to compute differentially private statistics it will be convenient to define
\[
{\hat{U}}_\ffeat(X, \rvidxs) =  U_\ffeat(X) + \sum_{i\in \rvidxs} \rvseq_i
\]
Many strategies have been proposed to compute such averages with higher or lower security and privacy guarantees.
The simplest but least secure is the classical trusted curator setting where all input is sent to a trusted party which makes the average and sends it back.
Strategies such as \cite{Shi2011,Bonawitz2017a,ChanSS12} focus on securely computing an average without disclosing the data to any party, assuming a honest-but-curious setting, i.e., parties are assumed to follow the protocol honestly even if they are curious and may try to infer information from what they observe.
The strategy proposed in \cite{Dwork2006ourselves} aim at better verifiability, but induces a cost quadratic in the number of parties.  The algorithms in \cite{Jayaraman2018,Sabater2021sample} integrate more strongly noise addition in the secure aggregation step, but relies on the additional assumption that two servers don't collude.
Recently, the shuffle model of privacy \cite{Cheu2019,amp_shuffling,Hartmann2019,Balle2020,Ghazi2020ICML} has been studied, where inputs are passed via a trusted/secure shuffler that obfuscates the source of the messages, leading to an alternative trust model.  It is also possible to distribute trust over the participating parties rather than relying on a limited number of servers for secret keeping \cite{Sabater2021accurate}.

We will here simply assume the existence of an operation $UStat(\ffeat,X,\rvseq:\rvidxs)$ that computes and publishes ${\hat{U}}_\ffeat(X,\rvidxs)$ in an appropriate way compliant with the considered attack model.  Accordingly, the cost of $UStat$ will depend on the strategy used, e.g., if the parties are assumed to be honest-but-curious the cost may be lower than if the method needs to be robust against certain malicious behavior.  We assume the operation $UStat$ has the necessary facilities to draw, store and keep secret value of the random variables in $\rvseq$.  As our generic algorithm will only use $U$-statistics and will perform subsequent calculations in the open, our approach will inherit its attack model directly from $UStat$.

\newcommand{\ecdf}[1]{{F_{{#1}}}}
\newcommand{\ecdfdp}[1]{{{\hat{F}}_{{#1}}}}
\newcommand{\ecdfsm}[1]{{{\acute{F}}_{{#1}}}}
\begin{definition}[Empirical cumulative distribution function]
  Let $\ffeat:\mathcal{X}\to\mathbb{R}$.
  Given a sample $X$, the empirical cumulative distribution function (ECDF) of $\ffeat$, denoted by $\ecdf{\ffeat}$, is the function $\ecdf{\ffeat}(X;\cdot) : \mathbb{R}\to[0,1]$ with \[\ecdf{\ffeat}(X,t) = \frac{1}{n}\left|\left\{i\in[n]\mid \ffeat(x_i)\le t\right\}\right|\]
\end{definition}

\section{Method}
\label{sec:method}
\subsection{Private ECDF}
\newcommand{\phimin}{\phi^{\hbox{{\small{min}}}}}
\newcommand{\phimax}{\phi^{\hbox{{\small{max}}}}}
Let $\ffeat:\mathcal{X}\to\mathbb{R}$.
Let $\phimin=\min_{i\in[n]} \ffeat(x_i)$ and $\phimax=\max_{i\in[n]} \ffeat(x_i)$.
Let $N\in \mathbb{N}$ and let $\tau \in \mathbb{R}^{[N]}$ be an ordered set of points on which one may want to evaluate $\ecdf{\ffeat}$, e.g., one may take $\tau=\{t\in\psi\mathbb{Z}\mid \phimin\le t \le \phimax\}$ for some precision parameter $\psi$, or if the relative error is more important than the absolute error one may take $\tau = \{e^{\psi z} \mid z\in\mathbb{Z} \wedge \phimin\le e^{\psi z} \le \phimax\}$.  We assume that $\tau_1=\phimin$ and $\tau_N=\phimax$.

\newcommand{\clN}{L} 
Let $\clN=\lceil\log_2(N)\rceil$.
Let $\rvidxL{L} = \{(j,l) \mid l\in[0,L]\wedge j\in [\lceil 2^{L-l}\rceil]\}$.
Let $\eta=\left(\eta_{j,l}\right)_{(j,l)\in \rvidxL{L}}$ be a set of random variables with  $\eta_{j,l} \sim Lap((L+1)/\epsilon)$ for $(j;l)\in\rvidxL{L}$.
Then, we define the function $\ecdfdp{\ffeat}$ by
\begin{equation}
  \label{eq:def.ecdfdp.pt}
  \ecdfdp{\ffeat}(X,\tau_i) = \ecdf{\ffeat}(X,\tau_i) + \frac{1}{|X|} \sum_{l=0}^L \eta_{\lceil i/2^l\rceil,l}
\end{equation}
While it is possible to reduce the errors obtained later by up to $15\%$ by using a base different from $2$ following \cite{DBLP:journals/corr/abs-2103-16787}, using base $2$ simplifies our explanation and the later algorithms.

\newcommand{\thmEcdfDp}{
  Publishing $\ecdfdp{\ffeat}(X,\tau_i)$ for all $i\in[N]$
  is $\epsilon$-DP. The expected squared error is $\mathbb{E}[(\ecdf{\ffeat}(x)-\ecdfdp{\ffeat}(x))^2]=2(L+1)^3/\epsilon^2$.
  }

\begin{theorem}
  \label{thm:ecdf.dp}
  \thmEcdfDp
\end{theorem}

The proof can be found in Appendix \ref{sec:proof.thm.ecdf.dp}.
It follows to a large extent the ideas from \cite{dwork2010differential,DBLP:journals/corr/abs-2103-16787}, but achieves a slightly better result by explaining the difference between outputs for adjacent datasets by not only positive but also negative changes in the noise terms.
Our ideas also allow for improving Theorem 1 in  \cite{DBLP:journals/corr/abs-2103-16787} on differentially private continual observation of statistics in data streams, we provide more details in appendix.

\subsection{Smooth DP ECDF}

\newcommand{\subtau}{B}
\newcommand{\lossFunc}{\mathcal{L}}
While we know that $\ecdf{\ffeat}$ is a non-decreasing function, due to the noise addition this doesn't hold anymore for $\ecdfdp{\ffeat}$.  We can correct this problem by finding the non-decreasing function $\ecdfsm{\ffeat}$ which minimizes $\lossFunc(\ecdfdp{\ffeat},\ecdfsm{\ffeat})$ for an appropriate loss function $\lossFunc$.   As $N$ may be large, in practice we may only be interested in finding an appropriate $\ecdfsm{\ffeat}$ for a limited set of points $(\tau_i)_{i\in \subtau}$ with $\subtau \subseteq [N]$.

In particular, we define $\ecdfsm{\ffeat}$ the following optimization problem:
\begin{equation}
  \label{eq:smooth.opt}
  \begin{array}{l}
    \text{minimize} \sum_{(i,j)\in\rvidxL{L}} \nu_{i,j}^p \\
    \text{s.t. }\forall i\in\subtau: \ecdfsm{\ffeat}(X,\tau_i) = \ecdfdp{\ffeat}(X,\tau_i) + \sum_{l=0}^L \nu_{\lceil i/2^l\rceil,l} \\
    \phantom{\text{s.t. }}\ecdfsm{\ffeat}(X,\min(B)) \ge 0 \\
    \phantom{\text{s.t. }}\ecdfsm{\ffeat}(X,\max(B)) \le 1 \\
    \phantom{\text{s.t. }}\forall i,j \in B: i<j \Rightarrow \ecdfsm{\ffeat}(X,\tau_i) \le \ecdfsm{\ffeat}(\tau_j) \\
    \end{array}
\end{equation}

Both $p=1$ and $p=2$ are plausible here.   As the loglikelihood of a vector of Laplace noise variables is proportional to its $1$-norm, $p=1$ may be appealing.  Still, we also will see a number of applications where the $2$-norm (i.e., $p=2$) is more appropriate.

Computing $\ecdfsm{}$ from $\ecdfdp{}$ is a post-processing step after achieving differential privacy, so it doesn't reduce the privacy guarantee.  On the other hand, it makes later processing requiring a non-decreasing function possible and may reduce the error induced by the DP noise (see Section \ref{sec:exp.smooth}).

\section{Algorithms}
\label{sec:algo}

\newcommand{\ffeatLE}[1]{{\ffeat_{\le {#1}}}}
In this section we will investigate algorithms to compute ECDFs and their inverse.  First, observe that we can write an ECDF evaluation as a U-statistic evaluation:
\[
\ecdf{\ffeat}(X,t) = U_{\ffeatLE{t}}(X)
\]
where $\ffeatLE{t}(x) = \indicVal{\ffeat(x)\le t}$.  In other words, if we want to evaluate $\ecdf{\ffeat}(X,t)$, we ask for each instance $x\in X$ whether $\phi(x)$ is smaller than $t$, and then count the positive answers.  Similarly, for the differential private version we can write
\begin{equation}
  \label{eq:ecdfdp.ustat}
\ecdfdp{\ffeat}(X,\tau_i) = {\hat{U}}_{\ffeatLE{\tau_i}}(X,\rvidxLI{L}{i})
\end{equation}
where $\rvidxLI{L}{i}=\{(\lceil i/2^l\rceil,l)\mid l\in[0,L] \}$.

\subsection{Pointwise evaluation}
Starting from any secure aggregation protocol implementing ${\hat{U}}_\cdot(\cdot,\cdot)$, Eq \eqref{eq:ecdfdp.ustat} allows for evaluating $\ecdfdp{\ffeat}$ on a number $\tau_i$ with $i\in[N]$ with the same security and privacy guarantees as that basic secure aggregation protocol.
If one wants to evaluate $\ecdfdp{\ffeat}$ on a set of values $\{\tau_i\}_{i\in \subtau}$ with $\subtau\subseteq [N]$, then we can just repeat the secure aggregation protocol.  In each iteration all parties must participate, which implies the communication cost is $O(n|\subtau|)$.

\subsection{Function secret sharing}
While this scheme is generic as it allows for any secure aggregation protocol, it is possible under particular attack models to improve on its complexity.  In particular, while protocols based on homomorphic encryption, multi-party computing or secret sharing may allow to aggregate values without revealing them, in their basic form they need a contribution of (and hence communication with) the parties storing the data in cleartext for each new query which must be answered, in our case for each $x$ on which we want to evaluate $\ecdfdp{\ffeat}(x)$ using a secure aggregation.  Recently, function secret sharing (FSS) techniques were proposed which allow to secretly share complete functions.  Here, we will use FSS for comparison functions, which was proposed in \cite{Boyle2015}.

Function secret sharing protocols consist of two operations: $gen$ and $eval$.  Given a function $f$ from the appropriate class, $gen(f)$ returns a set of $m\ge 2$ keys $(k^{(1)} \ldots k^{(m)})$.  These keys are indistinguishable (for algorithms running in polynomial time) from randomly drawn keys and can therefore be distributed over servers which are trusted to not collude.  To evaluate the function on some input $x$, the server who received $k^{(j)}$ ($j\in[m]$) can apply $y^{(j)}=eval(k^{(j)},x)$. These outputs $y^{(j)}$ are still indistinguishable from numbers drawn from some random distribution, but have the property that $f(x)=\sum_{j=1}^m y^{(j)}$.

\newcommand{\ffeatGE}[1]{{\ffeat_{\ge {#1}}}}
In our case, we exploit the class of comparison functions, i.e., the class of functions $\mathcal{F}^<_N = \{\ffeatGE{\tau_i} \mid i \in [N] \}$ where $\ffeatGE{t}(x) = \indicVal{\ffeat(x)\ge t}$. The work \cite{Boyle2015} proposes functions $gen$ and $eval$ for this class.  The $gen$ function returns a pair of keys, but can be extended to returning a larger number $m>2$ of keys.  Every party $P_i$ with private data $x_i$ applies the $gen$ function to $\ffeatGE(x_i)$, i.e., sets $(k_i^{(1)},k_i^{(2)}) = gen(\ffeatGE(x_i))$.  For $j\in [m]$, server $j$ receives all $j$-th keys, i.e., the set $K^{(j)}=\{k_i^{(j)} \mid i\in [n]\}$.  Whenever one wants to evaluate the ECDF at some point, all servers $j\in[m]$ compute $Y^{(j)}=\sum_{i=1}^n eval(k_i^{(j)})$, and then jointly sum $\sum_{j\in[m]} Y^{(j)}$ and add the appropriate DP noise.  The length of the generated keys is $O(L(\lambda+\log(n)))$ where $\lambda$ is a security parameter typically larger than $\log(n)$.

The cost of this algorithm can be divided in two phases.  First, the preprocessing phase where keys are generated is dominated by the sending of a key from each data owner party to each server. Second, the evaluation phase is relatively cheap, for every evaluation the user sends to all servers the number $\tau_i$ on which to evaluate the ECDF, the servers communicate among themselves for the addition and then send the answer back to the user.  The communication cost is hence linear in $m$.  The computation cost involves all servers going over the key with the input $\tau_i$, and hence the computation cost is linear in the key length for every server.

The size of the keys $k_i$ is constant but considerable (typically a few kilobits) so this approach is not recommended for small $n$ or $N$, but the communication cost is constant for the parties owning the data and only linearly in the number of evaluations $|\subtau|$ for the $m$ servers, giving a total communication cost of $O(n+m|B|)$ which is for a constant $m$ asymptotically better than the generic approach with the additional advantage that after the initial distribution of the keys only the $m$ servers need to stay online to answer queries during the computation.

\subsection{Inverse ECDF evaluation}

Next to evaluating an ECDF, one also often needs to evaluate the inverse ECDF $\ecdf{\ffeat}^{-1}$, i.e., one would like to know what is the value corresponding to a particular quantile.  A natural strategy is to apply binary search, as illustrated by Algorithm \ref{algo:invEcdf}.

\begin{algorithm}[ht]
\caption{}\label{algo:invEcdf}
  \begin{algorithmic}
    \Function{Inv-Ecdf-DP}{$p$}
      \State{$a \gets 0$; $b\gets 1$}
      \While{$b-a > \psi$}  \Comment{$\psi$ = precision}
        \State{$m\gets (a+b)/2$} \Comment{Consider middle}
        \If{$\ecdfsm{\ffeat}(X,m) < p$} \Comment{Split interval}
          \State{$a\gets m$}
        \Else
          \State{$b\gets m$}
        \EndIf
        \EndWhile \Comment{Until interval small enough}
      \State{\Return{$(a+b)/2$}}
    \EndFunction
  \end{algorithmic}
  \end{algorithm}

\section{Applications}
\subsection{The ROC curve and the area under it}
\label{sec:app.roc}

The ROC curve \cite{Provost98b:proc} is a popular way to visualize the characteristics of a classifier.  It gives a more complete view than a single performance measure such as accuracy or the area under the ROC curve.

\newcommand{\trueclass}{c^*}
\newcommand{\estscore}{c_0}
\newcommand{\estclass}{{\hat{c}}}

As before, let $\mathcal{X}$ be a space of instances.  Let $\trueclass:\mathcal{X}\to\{0,1\}$ be a function assigning to each instance $x$ its true class label.  Let $\estscore:\mathcal{X}\to \mathbb{R}$ be a function assigning to each instance an estimated score, where instances with a lower score are more likely to be positive (have class $1$) and instance with a higher score are more likely to be negative (have class $0$).  Let $\estclass(t, x) = \indicVal{\estscore(x)\le t}$.

Let $X\in\mathcal{X}^n$ be a dataset and let $t$ be a threshold, the true positive rate is
\[
  TPR(X,t) = \frac{\{x\in X\mid \trueclass(x)=1 \wedge \estclass(t, x)=1\}}{\{x\in X\mid \trueclass(x)=1\}}
\]
and the false positive rate is
\[
  FPR(X,t) =  \frac{\{x\in X\mid \trueclass(x)=0 \wedge \estclass(t, x)=1\}}{\{x\in X\mid \trueclass(x)=0\}}
\]
The ROC curve plots TPR against FPR, so $(r_F,r_T) \in ROC$ iff there is a threshold $t$ such that $r_T=TPR(X,t)$ and $r_F=FPR(X,t)$.
\newcommand{\cmax}{c_{max}}
Let $\cmax=\max_{x\in X} \estscore(x)$.
Let $\phi_{TP}(x) = \indicVal{\trueclass(x)=1 \wedge \estclass(t, x)=1}$ and $\phi_{FP}(x) = \indicVal{\trueclass(x)=0 \wedge \estclass(t, x)=1}$.  Then, $TPR(X,t) = \ecdf{\phi_{TP}}(X,t)/\ecdf{\phi_{TP}}(X,\cmax)$ and $FPR(X,t) = \ecdf{\phi_{FP}}(X,t)/\ecdf{\phi_{FP}}(X,\cmax)$.  So publishing $\ecdfdp{\phi_{TP}}$ and $\ecdfdp{\phi_{FP}}$ is sufficient to transmit an approximate ROC curve.  Moreover, if we add sufficient noise to make both functions $\epsilon/2$-differentially private, their combined disclosure is $2\epsilon$-differentially private.

Figure \ref{fig:ex.roc.framingham.0.5} illustrates this process on a relatively small dataset (see Section \ref{sec:exp.dataset}) where the effect of DP noise is clearly visible.  As both coordinates of a point in the ROC curve are ECDFs to which noise is added, one can see that $\ecdfdp{\ffeat}$ has both horizontal and vertical deviations from the true ROC curve.  In this figure, one can also see to some extent the organization of the noise as binary tree: the first half of the DP curve seems to be above the true ROC curve, while the latter half is lower, suggesting the noise variable that was added to the first half got a clearly higher value.  The smoothed curves resolve this problem and stay in this case much closer to the true ROC curve.

\begin{figure}[ht]
\vskip 0.2in
\begin{center}
\centerline{\includegraphics[width=\columnwidth]{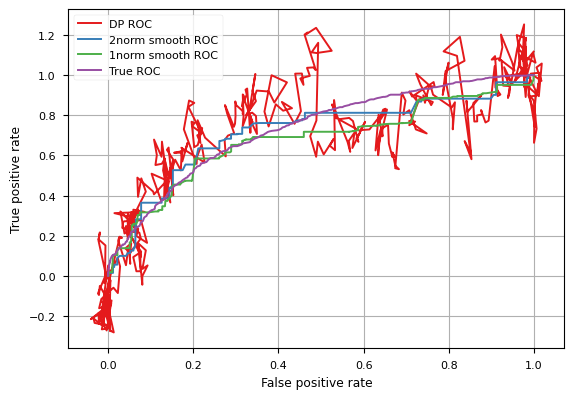}}
\caption{ROC curve for logistic regression on the Heart disease dataset, and $\epsilon$-DP curves with $\epsilon=0.5$.}
\label{fig:ex.roc.framingham.0.5}
\end{center}
\vskip -0.2in
\end{figure}

As the size of the dataset increases, the impact of the noise decreases and even for smaller values of $\epsilon$ the differentially private curve gives a good picture of the true one, e.g., see Figures \ref{fig:ex.roc.bank.0.2} and  \ref{fig:ex.roc.diabetic.0.05} in appendix.

\subsection{Calibration and the Hosmer-Lemeshow statistic}
\label{sec:app.HL}

The Hosmer-Lemeshow test is a statistical test which is popular in health science \cite{HosmerLemeshow2013} and other domains.  While it has some limitations, there is no universal agreement on what is the best alternative, and the Hosmer-Lemeshow statistic is still quite commonly used.
It is often used as calibration test for logistic regression, but may also be applied to other machine learning models \cite{Meyfroidt2011}.

While applying a non-decreasing function to the output of a classifier will not change its ROC curve, it will impact its callibration.  For models outputting a probability that an instance is positive, it is desirable that the estimated probability of being positive is close to the true probability.  The Hosmer-Lemeshow statistic can evaluate such goodness-of-fit.
To compute it, a first step is to rank all instances in increasing order of predicted probability of being positive.  Then, this ranked list is partitioned into $Q$ equally sized quantile groups, where typically $Q=10$, so group $1$ contains the instances which are predicted to be most likely negative and group $Q$ contains the instances which are predicted to be most likely positive.
Then, the Hosmer-Lemeshow statistic is defined by
\begin{equation}
  \label{eq:HL.def}
  H=\sum\limits_{i=1}^{Q}\bigg(\frac{(O_{1i}-E_{1i})^2}{E_{1i}}+\frac{(O_{0i}-E_{0i})^2}{E_{0i}}\bigg)
  \end{equation}
Where $O_{1i}$ is the observed number of positive instances, $E_{1i}$ is the predicted number of positive instances, $O_{0i}$ is the observed number of negative instances  and $E_{0i}$ is the predicted number of negative instances in group $i$.
So a group is a collection of instances where the expected probability lies in a specific interval, and the test checks how far the observed class distribution deviates from this.
Once the Hosmer-Lemeshow statistic is computed, one can compare it to a chi-squared distribution with $Q-2$ degrees of freedom to test the hypothesis that the observed classes in each group are distributed according to the predictions.

\newcommand{\HLModelFunc}{\mathcal{M}}
\newcommand{\HLObserveFunc}{\mathcal{Y}}
Assume we have a model $\HLModelFunc:X\to[0,1]$ mapping each instance on the predicted probability it is positive and a function $\HLObserveFunc:X\to\{0,1\}$ mapping every instance on its true class label, either $0$ (negative) or $1$ (positive).  Let $\HLModelFunc[l,u](x) = \HLModelFunc(x).\indicVal{l\le \HLModelFunc(x)\le u}$ and $\HLObserveFunc[l,u](x) = \HLObserveFunc(x).\indicVal{l\le \HLModelFunc(x)\le u}$.
Then, $H$ can be computed using only $U$-statistics following Algorithm \ref{algo:hlstat}.

\begin{algorithm}[ht]
\caption{}\label{algo:hlstat}
\begin{algorithmic}[1]
  \Function{HL-Stat-DP}{$X$ : dataset, $Q$ : number of groups, $\epsilon$ : privacy level; $L$ : precision parameter}
  \State{$t_0\gets 0$; $t_Q\gets 1$; $n\gets |X|$}
  \State{$\epsilon'\gets \epsilon/(L+9)$}
  \ForAll{$q\in[Q-1]$:}
  \State{\label{ln:HL.ecdf}$t_q\gets \ecdfdp{\HLModelFunc}^{-1}(q/Q)$} \Comment{$(L+1)\epsilon'$-DP}
  \EndFor
  \ForAll{$q\in[Q]$, $s \in \{0,1\}$}
  \State{\label{ln:HL.expobs1}$E_{s,i} \gets n.{\hat{U}}_{\HLModelFunc[t_{q-1},t_q]}(X, \{(-q,s)\})$} \Comment{$\epsilon'$-DP}
  \State{\label{ln:HL.expobs2}$O_{s,i} \gets n.{\hat{U}}_{\HLObserveFunc[t_{q-1},t_q]}(X, \{(-q,s)\})$} \Comment{$\epsilon'$-DP}
  \EndFor
  \State{\Return{$\sum\limits_{i=1}^{Q}\bigg(\frac{(O_{1i}-E_{1i})^2}{E_{1i}}+\frac{(O_{0i}-E_{0i})^2}{E_{0i}}\bigg)$}}
  \EndFunction
\end{algorithmic}
\end{algorithm}

In appendix \ref{sec:proof.HL.DP}, we prove the following theorem:
\newcommand{\thmHLDP}{Running Algorithm \ref{algo:hlstat} and disclosing any results of $U$-statistics it invokes is $\epsilon$-DP.}
\begin{theorem}
  \label{thm:HL.DP}
  \thmHLDP
\end{theorem}
In particular, Algorithm \ref{algo:hlstat} may publish both the ECDF of predicted probabilities $\ecdfdp{\HLModelFunc}$ and the statistics $O_{s,i}$ and $E_{s,i}$, we only assume the aggregation primitives used for computing the $U$-statistics are secure.

  \section{Experiments}
  \label{sec:exp}

  \subsection{Setup}

  Our experiments aim at providing illustrations to our approach and at providing more insight in the practical behavior of our proposal.
Unless stated otherwise, our experiments average over 100 runs.
Experiments were performed on Intel Core i7-4600U CPUs at 2.10GHz with 16Gb of RAM.
 Code (using Python 3.9) to reproduce all experiments will be downloadable from the website of the authors. 
 For experiments involving ROC curves and the Hosmer-Lemeshow statistic, we first trained a logistic regression model using the Scikit-Learn package.  Our goal was not to obtain the most performant classifier, but to illustrate how our methods can assess properties of an arbitrary classifier.

  \subsection{Datasets}
  \label{sec:exp.dataset}
  \subsubsection{Synthetic data}

  We will use a dataset $X_{pois(\lambda)}$ constructed as follows.  Let $N=2^{15}$.  For all values $i\in[N]$, the number of instances in $X_{pois(\lambda)}$ such that $\ffeat(x_j)=i$ follows a Poisson distribution $Pois(\lambda)$.  At a high level this dataset has a steadily increasing ECDF, however locally there is quite some variance in how quickly it increases.

  \subsubsection{Real-world data}

We use the following datasets:
\begin{itemize}
\item Heart disease prediction : This data set aims to pinpoint the most relevant/risk factors of heart disease as well as predict the 10-year risk of coronary heart disease using information about the individuals (e.g. age, gender, smoking habits and blood pressure).
  (\url{https://www.kaggle.com/dileep070/heart-disease-prediction-using-logistic-regression})
\item Bank full : This data set is related to a marketing campaign of a Portuguese banking institution. It intends to know if the client would subscribe or not to the product (bank term deposit), using features about the clients (e.g. age, job, education level, marital status, owning a house, having loans).  (\url{https://www.kaggle.com/krantiswalke/bankfullcsv})
\item Diabetes data set: It studies the effects of diabetes to the readmission of patients to the hospital (\url{https://archive.ics.uci.edu/ml/datasets/Diabetes+130-US+hospitals+for+years+1999-2008\#})
\end{itemize}
Table \ref{tab:datasets} summarizes some relevant characteristics.

\begin{table}[t]
\caption{Dataset summary giving the number of instances, the fraction of positive instances and the number of features.}
\label{tab:datasets}
\vskip 0.15in
\begin{center}
\begin{small}
\begin{sc}
\begin{tabular}{lrrr}
  \hline
Data set & \#inst & pos.frac & \#feature \\\hline
Heart disease & 4328 & 0.152 & 15 \\
Bank-full & 45211 & 0.117 & 16 \\
Diabetes & 101766 & 0.460 & 49 \\
\end{tabular}
\end{sc}
\end{small}
\end{center}
\vskip -0.1in
\end{table}

\subsection{Smoothing}
\label{sec:exp.smooth}

To investigate the effect of smoothing on the error induced by the DP noise, we start from the $X_{Pois(\lambda)}$ dataset.  Figure \ref{fig:smooth.mseReduc.epsilon} shows for $p=1,2$ curves plotting as a function of $\epsilon$ the effect of smoothing $\ecdfdp{\ffeat}$ on the error made by the DP noise as \[
\frac{\left\|\ecdf{\ffeat}(X_{pois(\lambda)},\tau)-\ecdfsm{\ffeat}(X_{pois(\lambda)},\tau)\right\|_2^2}{\left\|\ecdf{\ffeat}(X_{pois(\lambda)},\tau)-\ecdfdp{\ffeat}(X_{pois(\lambda)},\tau)\right\|_2^2}
\]
One can see that, as expected for a $2$-norm evaluation, the $2$-norm smoothing outperforms the $1$-norm smoothing, especially for small $\epsilon$ values.  For roughly $\epsilon \ge 0.2$ its curve is below $1$, meaning $2$-norm smoothing in this range reduces the $2$-norm error induced by the DP noise.  The higher $\epsilon$ becomes, the less likely it is the small amount of DP noise will make the ECDF non-increasing, hence the effect of smoothing on the DP error goes to zero.  Figure \ref{fig:smooth.mseReduc.lambda} shows the same information for fixed $\epsilon$ and varying $\lambda$.  For large $\lambda$, the ECDF is strongly increasing and the relevance of smoothing is limited as adding DP noise rarely makes it non-decreasing. For moderate values of $\lambda$, smoothing improves the DP noise induced error.

\begin{figure}[ht]
\vskip 0.2in
\begin{center}
  \centerline{\includegraphics[width=\columnwidth]{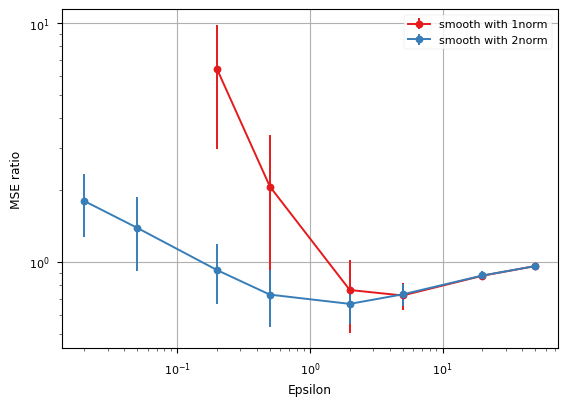}}
\caption{Effect of smoothing on DP error - fixed $\lambda=3$}
\label{fig:smooth.mseReduc.epsilon}
\end{center}
\vskip -0.2in
\end{figure}

\begin{figure}[ht]
\vskip 0.2in
\begin{center}
\centerline{\includegraphics[width=\columnwidth]{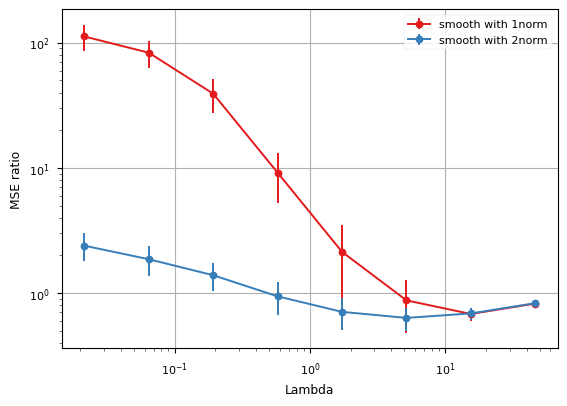}}
\caption{Effect of smoothing on DP error - fixed $\epsilon=1$}
\label{fig:smooth.mseReduc.lambda}
\end{center}
\vskip -0.2in
\end{figure}

\subsection{Evaluating an ECDF or its inverse}

Starting from the $X_{Pois(\lambda)}$ dataset, Figure \ref{fig:invEcdf} plots as a function of $\epsilon$ the mean square errors $(\ecdf{\ffeat}(x)-\ecdfdp{\ffeat}(x))^2$, $\ecdf{\ffeat}^{-1}(x)-\ecdfdp{\ffeat}^{-1}(x)$ and $\ecdf{\ffeat}^{-1}(x)-\ecdfsm{\ffeat}^{-1}(x)$ when evaluating on points $x$ uniformly distributed over the relevant domains.  The inverse ECDF evaluations are performed using Algorithm \ref{algo:invEcdf}.

\begin{figure}[ht]
\vskip 0.2in
\begin{center}
\centerline{\includegraphics[width=\columnwidth]{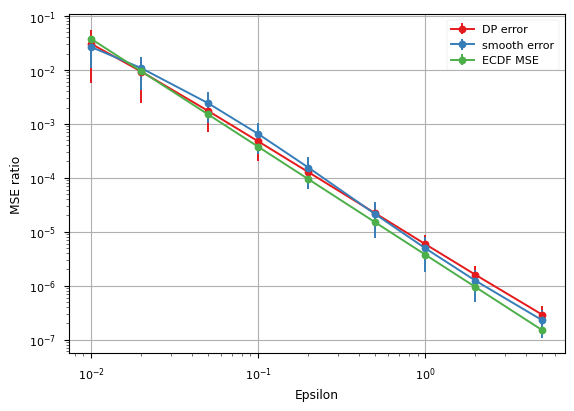}}
\caption{Inverse ECDF}
\label{fig:invEcdf}
\end{center}
\vskip -0.2in
\end{figure}

The fact that $\ecdfdp{\cdot}$ is not guaranteed to be non-decreasing and this could confuse the binary search algorithm doesn't seem to have a strong impact on the accuracy of the evaluation.  In fact, both the mean squared errors of $\ecdfdp{\ffeat}$ and the inverses of its smoothed and non-smoothed versions are of the same order of magnitude.

\subsection{ROC curve estimation}

Figure \ref{fig:roc.symmdiff} plots, as a function of $\epsilon$, the $1$-norm difference between the true ROC curve and the smoothed differentially private ROC curve for the Bank dataset.  This corresponds to the area of the symmetric difference of the area under the true ROC curve and the area under the differentially private ROC curve.  Even when the area under both curves would be the same, this value can be non-zero as the curves themselves differ.  As expected, error decreases with increasing $\epsilon$.  We see that the $1$-norm and $2$-norm smoothing perform about equally well.  As the unsmoothed differentially private ROC curves cross themselves, it is hard to compare their area with the area of the true ROC curve.

\begin{figure}[ht]
\vskip 0.2in
\begin{center}
\centerline{\includegraphics[width=\columnwidth]{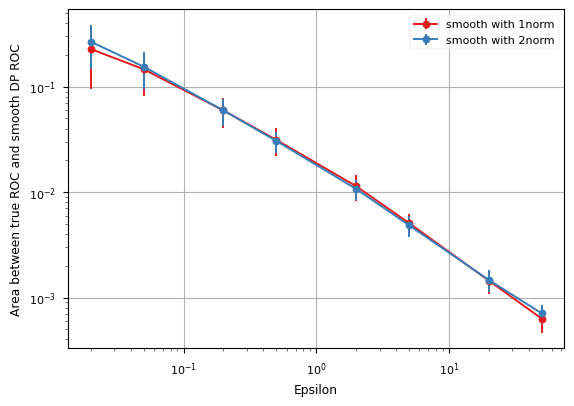}}
\caption{ROC curve estimation error}
\label{fig:roc.symmdiff}
\end{center}
\vskip -0.3in
\end{figure}

\subsection{Hosmer-Lemeshow}

As for ROC curves, to understand calibration one can both look at the complete picture of the predicted and observed counts of positive / negative instances in each of the $Q$ groups, and one can look at the HL-statistic which summarizes it as a $\chi^2$ statistic.  Here, we show results for the latter approach.

Figure \ref{fig:HL.bank} shows for a logistic regression model on the Bank dataset the mean square error of the Hosmer-Lemeshow test when computed privately for different values of $\epsilon$.  In appendix, Figure \ref{fig:HL.dia} provides the same information for the Diabetes dataset.  Especially for the Bank dataset we observe a quite large variance over the several runs.  This can be explained by the fact that the Bank dataset has less balanced classes (see Table \ref{tab:datasets}).  This causes both predicted and observed counts of positive examples in especially the lowest of the $Q=10$ groups to be rather small.  Even a small error in the small number $E_{1,1}$ (the predicted number of positives in the first group) appearing in the denominator of a term in the $H$ statistic (see Eq \ref{eq:HL.def}) may cause a large error in the final statistic.  We can conclude that especially if $\epsilon$ is small, if there is no need to know the individual statistics $E_{\cdot,\cdot}$ and $O_{\cdot,\cdot}$ and multi-party computation is available for other operations than $U$-statistics, a direct approach may be preferable where one first computes securely the correct statistic and then adds noise at the end proportional to the sensitivity of (only) the HL statistic.

\begin{figure}[ht]
\vskip 0.2in
\begin{center}
\centerline{\includegraphics[width=\columnwidth]{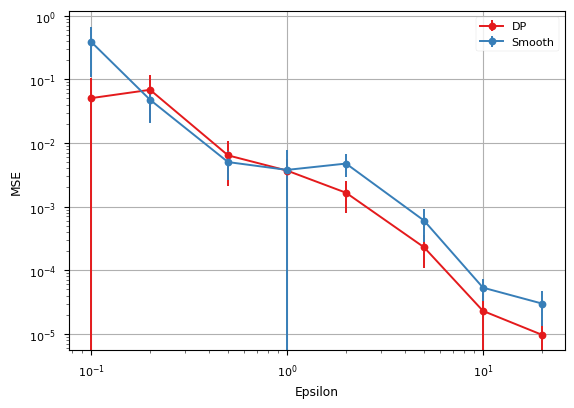}}
\caption{Hosmer-Lemeshow statistic relative MSE for a logistic regression model on the Bank dataset}
\label{fig:HL.bank}
\end{center}
\vskip -0.3in
\end{figure}

\subsection{Runtime}

The cost of a distributed algorithm is often dominated by its communication cost.  These costs have been analyzed in Section \ref{sec:algo}.  Here, we complement this analysis with an experiment on the most expensive local computation: the smoothing of the private ECDF.  For solving the optimization problem in Eq \ref{eq:smooth.opt}, we use the \textsc{cvxopt} package, in particular a linear program solver for the $1$-norm smoothing and a quadratic program solver for the $2$-norm smoothing.  Figure \ref{fig:runtime.smooth} shows the runtime as a function of $N$, which is a good problem size parameter as the number of variables in the optimization problem grows as $O(N\log(N))$.  One can observe that the quadratic program is solved more quickly.

\begin{figure}[ht]
\vskip 0.2in
\begin{center}
\centerline{\includegraphics[width=\columnwidth]{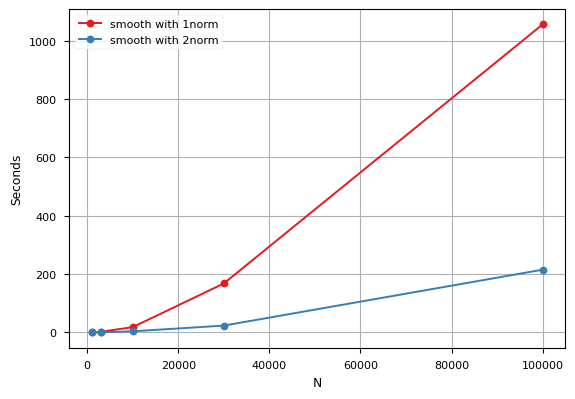}}
\caption{Runtime of solving Eq \ref{eq:smooth.opt}}
\label{fig:runtime.smooth}
\end{center}
\vskip -0.3in
\end{figure}

\section{Discussion}
\label{sec:concl}

In this paper we studied differentially private empirical cumulative distribution functions.  We proved privacy guarantees and proposed algorithms to securely compute such private ECDFs.  We elaborated in more depth two applications of ECDFs: ROC curves and the Hosmer-Lemeshow statistic.  Our experimental results suggest the approach can convey the full information of an ECDF at a reasonable precision and privacy level.

Cumulative distribution functions are important in various other areas of machine learning and statistics.  One application we didn't elaborate in-depth concerns histograms of (discretized) continuous variables, e.g., a histogram of the yearly income of a set of persons grouped in bins of $\$5000$.  A common strategy \cite{DBLP:journals/corr/abs-2103-16787} is to add DP noise to the count in each bin independently.  An alternative strategy would be to consider the cumulative distribution, which can be made private by noise only logarithmic in the number of bins.  The interpretation then is that noise can not only consist of a change in the count in a bin, but also in a shift from one bin to an adjacent one.

There are several potentially interesting lines of future work.  Among others it would be interesting to elaborate more applications of ECDF, to develop more efficient algorithms to securely compute private ECDF and get a better understanding of the various statistical processes affecting the error DP noise induces.

\newpage
\appendix

\bibliography{biblio}
\bibliographystyle{plain}

\section{Proofs}

\subsection{Proof of Theorem \ref{thm:ecdf.dp}}
\label{sec:proof.thm.ecdf.dp}

Before proving Theorem \ref{thm:ecdf.dp}, we first introduce some additional definitions and lemmas.
For $L\in\mathbb{N}\setminus\{0\}$, let
\[Z_{L,d} = \left(\indicSet{[d+(j-1)2^l+1,d+j2^l]}\right)_{(j,l)\in \rvidxL{L}}\]
be a vector of functions indexed by $\rvidxL{L}$ where $\indicSet{X}:\mathbb{N}\to\{0,1\}$ is a function with $\forall x\in X: \indicSet{X}(x)=1$ and $\forall x\in\mathbb{N}\setminus X : \indicSet{X}(x)=0$.  Note that $[d,d-1]=\emptyset$ and hence $\indicSet{[d,d-1]}=0$.

\begin{lemma}
  \label{lm:need.halfL.noiseterms}
  Let $L\in \mathbb{N}\setminus\{0\}$, $d\in\mathbb{N}$ and $b-d\in\left[2^L\right]$, then there exist vectors $\chi^s\in\{-1,0,1\}^{\rvidxL{L}}$, $s\in\{+,-\}$, with the number of non-zero elements $\|\chi^s\|_0$ at most $\lceil (L+1)/2\rceil$ such that $\indicSet{[d+1,b]} = \chi^-.Z_{L,d}$ and  $\indicSet{[b,d+2^L]} = \chi^+.Z_{L,d}$.
\end{lemma}
\begin{proof}
  We need to prove that we can write the functions $\indicSet{[d+1,b]}$ and $\indicSet{[b,d+2^L]}$ as weighted sums (with coefficients $-1$ or $+1$) of elements of $Z_{L,d}$.  We proceed by induction.

  \textit{Base cases.} Consider first $L=0$ and $L=1$.
In both cases it is easy to verify that $\indicSet{[d+1,b]}$ and $\indicSet{[b,d+2^L]}$ are both elements of $Z_{2,d}$ and hence $\chi^+$ and $\chi^-$ can be set to suitable base vectors, i.e., $\|\chi^+\|_0=\|\chi^-\|_0=1$.

  \textit{Induction step.}
  Suppose that $L\ge 2$ and assume that the lemma has been proven for $L'\le L-2$.  There are four cases, depending on $b'=\lceil (b-d)/2^{L-2}\rceil\in[4]$.
  \vspace{-1mm}
  \begin{itemize}
    \item
      Consider first $b'=0$, i.e., $b-d \le 2^{L-2}$:
      \begin{itemize}
        \item[($\chi^-$)]
          Applying the lemma for $L'=L-2$ we can write $\indicSet{[d+1,b]}$ as a weighted sum of $\lceil (L'+1)/2\rceil < \lceil (L+1)/2\rceil$ elements of $Z_{L',d}\subseteq Z_{L,d}$.
        \item[($\chi^+$)] As $\indicSet{[b,d+2^L]}=\indicSet{[d+1,d+2^L]}-\indicSet{[d+1,b-1]}$, we can write $\indicSet{[b,d+2^L]}$ as a weighted sum of at most $\lceil (L+1)/2\rceil$ elements of $Z_{L,d}$: $\indicSet{[d+1,d+2^L]} \in Z_{L,d}$ and either $\indicSet{[d+1,b-1]} = 0$ or $b-1\in[2^{L-2}]$ implying that $\indicSet{[d+1,b-1]}$ as in case ($\chi^-$) above.
          \end{itemize}
        \item Consider next $b'=1$, i.e., $2^{L-2}+1\le b-d\le 2^{L-1}$.
          \begin{itemize}
          \item[$\chi^-$] As  $\indicSet{[d+1,b]}=\indicSet{[d+1,d+2^{L-2}]}+\indicSet{[(d+2^{L-2})+1,b]}$,  $\indicSet{[b,d+2^L]}\in Z_{L,d}$ it suffices to apply the induction hypothesis and note that we can write $\indicSet{[(d+2^{L-2})+1,b]}$ as a sum of $\lceil (L+1)/2\rceil-1$ elements of $Z_{L-2,d+2^{L-2}}\subseteq Z_{L,d}$.
            \item[$\chi^+$] $\indicSet{[b,d+2^L]} = \indicSet{[b,d+2^{L-1}]} + \indicSet{[d+2^{L-1}+1,d+2^L]}$ and $\indicSet{[b,d+2^{L-1}]}\in Z_{L,d}$ hence it suffices to apply the induction hypothesis to $\indicSet{[d+2^{L-1}+1,d+2^L]}$.
          \end{itemize}
\item
  Cases $b'=2$ and $b'=3$ are analoguous to cases $b'=1$ and
   $b'=0$ respectively.
\end{itemize}
This completes the proof.
\end{proof}

\textbf{Theorem \ref{thm:ecdf.dp}} \thmEcdfDp

\begin{proof}
  Consider two adjacent datasets $X^{(1)}$ and $X^{(2)}$.
  As the datasets are adjacent, they differ in only one instance,
  i.e., there exists a dataset $X'$ and instances $x_\Delta^{(1)}$ and $x_\Delta^{(2)}$ such that $X^{(s)}=X'\cup\{x_\Delta^{(s)}\}$ for $s\in \{1,2\}$.

  The inner product of $\rvidxL{L}$-indexed vectors $\eta Z_{L,0}$ is a function mapping any $i\in[2^L]$ to
\[ \sum_{(j,l)} \{\eta_{j,l}\mid (j-1)2^l+1 \le i \le j2^l\}
  = \sum_{l=0}^L \eta_{\lceil i/2^l\rceil,l}.
\]
  Then, defining $\ecdf{\ffeat}(X,\tau) = \left(\ecdf{\ffeat}(X,\tau_i)\right)_{i\in[2^L]}$ and  $\ecdfdp{\ffeat}(X,\tau) = \left(\ecdfdp{\ffeat}(X,\tau_i)\right)_{i\in[2^L]}$, where we set $\forall i\in [N+1, 2^L]:\tau_i=\phimax$, we can rewrite Eq \eqref{eq:def.ecdfdp.pt} as
\begin{equation}
  \label{eq:def.ecdfdp.func}
  \ecdfdp{\ffeat}(X,\tau) = \ecdf{\ffeat}(X,\tau) + \frac{\eta Z_{L,0}}{n}
\end{equation}
For $s\in{1,2}$, there holds
\begin{equation}
  \label{eq:diff.ecdf.XpXs}
n\ecdf{\ffeat}(X^{(s)},\tau) - (n-1)\ecdf{\ffeat}(X^{(s)},\tau) = \indicSet{[t_s+1,2^L]}
\end{equation}
where $t_s=\max\left\{i \in [N]\mid x_\Delta^{(s)} > \tau_i\right\}$.
Without loss of generality we assume that $X_\Delta^{(2)} < x_\Delta^{(1)}$.
Combining Eq \eqref{eq:def.ecdfdp.func} and twice Eq \eqref{eq:diff.ecdf.XpXs} we get
\begin{equation}
  \label{eq:ecdfdp.diffX12}
  n\ecdfdp{\ffeat}(X^{(2)},\tau)
  = n\ecdf{\ffeat}(X^{(1)},\tau)
   +\indicSet{[t_1+1,t_2]}+\eta Z_{L,0} 
\end{equation}
\newcommand{\lmid}{{l_{\hbox{{\tiny{m}}}}}}
\newcommand{\jmid}{{j_{\hbox{{\tiny{m}}}}}}
Consider the largest $\lmid\in [0,L]$ for which there exists a $\jmid\in\mathbb{N}$ such that $t_1\le \jmid 2^{\lmid} < t_2$.  Notice that $\jmid$ is odd, as else we would have $(\jmid/2)2^{\lmid+1}=\jmid 2^{\lmid}$ and $\lmid$ would not be maximal.  Also, there holds $(\jmid-1)2^{\lmid} < t_1$ as else we would have $t_1 \le (\jmid-1)2^{\lmid}$ with $\jmid-1$ even and again $\lmid$ would not be maximal.  Similarly we can infer $t_2\le (\jmid+1)2^{\lmid}$. As $t_2\le 2^L$ there follows $\lmid\le L-1$.
Let $d_1=(\jmid-1)2^{\lmid}$ and $d_2 = \jmid 2^{\lmid}$.
Applying twice Lemma \ref{lm:need.halfL.noiseterms} we can conclude there exist vectors $\chi_1^+,\chi_2^- \in\{-1,0,+1\}^{\rvidxL{\lmid}}$ with $\|\chi_1^+\|_0\le \lceil L/2\rceil$ and $\|\chi_2^-\|_0\le \lceil L/2\rceil$ such that
$\chi_1^+.Z_{\lmid,d_1 }= \indicSet{[t_1+1,d_2]}$ and $\chi_2^-.Z_{\lmid,d_2 } = \indicSet{[d_2 +1,t_2]}$.  As the elements of the vectors $Z_{\lmid,d_1 }$ and $Z_{\lmid,d_2}$ also occur in $Z_{L,0}$, we can conclude that there exists a vector $\chi\in\{-1,0,+1\}^{\rvidxL{L}}$ with $\|\chi\|_0 \le L+1$ such that
\begin{equation}
  \label{eq:chi.t1t2}
  \chi.Z_{L,0} = \indicSet{\left[t_1+1,t_2\right]}.
  \end{equation}

We now express the probability of observing $\ecdfdp{\ffeat}(X^{(2)})$ given $X^{(2)}$ and compare it the probability of making the same observation given $X^{(1)}$.  Using Eqs \eqref{eq:ecdfdp.diffX12} and \eqref{eq:chi.t1t2} and setting $\Gamma(y,X^{(1)}) = y-n.\ecdfdp{\ffeat}(X^{(1)},\tau)$,
\begin{eqnarray*}
  \lefteqn{P\left(n.\ecdfdp{\ffeat}(X^{(2)},\tau)=y\right)}
  &&\\
  &=& P\left(n\ecdf{\ffeat}(X^{(1)},\tau)+\indicSet{[t_1+1,t_2]}+\eta Z_{L,0} = y\right) \\
&=& P\left(\chi Z_{L,0}+\eta Z_{L,0}=\Gamma(y,X')\right)\\
  &=& \int_u P(\eta = u) \, \indicVal{\chi Z_{L,0}+\eta Z_{L,0}=\Gamma(y,X')}\\
  &=& \int_u P(\eta = u+\chi) \, \indicVal{\eta Z_{L,0}=\Gamma(y,X')}
\end{eqnarray*}
Also,
\begin{eqnarray*}
  \lefteqn{P\left(n.\ecdfdp{\ffeat}(X^{(2)},\tau)=y\right)} &&\\
&=& \int_u P(\eta = u) \, \indicVal{\eta Z_{L,0}=\Gamma(y,X')}
\end{eqnarray*}
The probability ratio for a given $\eta$ is
  \begin{eqnarray*}
    \lefteqn{\left|\log\left(\frac{P(\eta = u+\chi)}{P(\eta = u)}\right)\right|}
    &&\\
    &=& \left|\sum_{(j,l)}  \log\left(\frac{P(\eta_{j,l} = u_{j,l}+\chi_{j,l})}{P(\eta_{j,l} = u_{j,l})}\right) \right|   \\
    &\le & \sum_{(j,l) : \chi_{j,l}\neq 0} \left|\log\left(\frac{P(\eta_{j,l} = u_{j,l}+\chi_{j,l})}{P(\eta_{j,l} = u_{j,l})}\right)\right|    \\
    &\le & (L+1) \frac{\epsilon}{L+1} \\
    &=& \epsilon
  \end{eqnarray*}
  We can conclude
  \[
\left|\log\left(\frac{P\left(n.\ecdfdp{\ffeat}(X^{(2)},\tau)=y\right)}{P\left(n.\ecdfdp{\ffeat}(X^{(2)},\tau)=y\right)}\right)\right| \le \epsilon
\]
as both probabilities are integrals over functions who only differ by a factor $e^\epsilon$.  This proves the privacy guarantee.

The expected squared error made by adding the noise is
\begin{eqnarray*}
  \lefteqn{\mathbb{E}\left[\left(\ecdfdp{\ffeat}(x)-\ecdf{\ffeat}(x)\right)^2\right]}
  && \\
  &=& \mathbb{E}\left[\left( \sum_{l=0}^L \eta_{\lceil i/2^l\rceil,l} \right)^2\right] \\
  &=& \sum_{l=0}^L \hbox{var}\left(\eta_{\lceil i/2^l\rceil,l}\right) \\
  &=& (L+1).2\left(\frac{L+1}{\epsilon}\right)^2 \\
  &=& 2(L+1)^3/\epsilon^2
  \end{eqnarray*}

\end{proof}

The above result has some implications for the more commonly studied problem of releasing statistics under continual observation \cite{DBLP:journals/corr/abs-2103-16787}.  In this problem, one considers data streams $(x_i)_{i=1}^N$ and wants to report at any time step $t$ the partial sum $s_t=\sum_{i=1}^t x_i$.  Two data streams are considered adjacent if they differ at most in one time step.  A change at a time step $t^*$ changes all partial sums between $t^*$ and $N$, i.e., the sums in a half-open interval.  In our setting, we also considered datasets adjacent if an instance changes its value, which means the partial sums change between its old value and its new value, i.e., the sums in a closed interval.   Applying our technique above to this problem, we get the following results:

\begin{theorem}
  \label{thm:improve.data.stream.dp}
  Let datasets $X^{(1)},X^{(2)}\in\mathcal{X}^N$ be adjacent if there exists at most one $i$ such that $X^{(1)}_i \neq X^{(2)}_i$.  Let $L=\lceil\log_2(N)\rceil$. Then, publishing ${\hat{s}}_t = \sum_{i=1}^t x_i + \sum_{l=0}^L \eta_{\lceil i/2^l\rceil,l}$ where $\eta_{j,l}\sim Lap(\lceil(L+1)/2\rceil/\epsilon$ is $\epsilon$-differentially private.    Similarly, with $\eta_{j,i}\sim \mathcal{N}(0,\lceil(L+1)/2\rceil z^2)$ the publishing is $1/2z^2$-zCDP (as defined in \cite{DBLP:journals/corr/abs-2103-16787}).
\end{theorem}

\begin{proof}
  The proof is a direct application of Lemma \ref{lm:need.halfL.noiseterms} using similar ideas as in the proof of Theorem \ref{thm:ecdf.dp}.
\end{proof}

The proof of Theorem 1 in \cite{DBLP:journals/corr/abs-2103-16787} concludes that at most $L$ terms in the partial sums they disclose will change between adjacent datasets, and hence need to compose $L$ differential private mechanisms.  Even though they focus on Renyi differential privacy rather than $\epsilon$-differential privacy, our idea allows in their setting too to reduce the number of changed terms with about a factor $2$ (when using base $2$) and hence to improve the privacy guarantee.

\subsection{Proof of Theorem \ref{thm:HL.DP}}
\label{sec:proof.HL.DP}

\textbf{Theorem \ref{thm:HL.DP}.} \thmHLDP
\begin{proof}
  The algorithm queries data at lines \ref{ln:HL.ecdf} and at lines  \ref{ln:HL.expobs1}--\ref{ln:HL.expobs2}.

  First, from Theorem \ref{thm:ecdf.dp} we know that by using $Lap(1/\epsilon')$ noise variables for evaluating $\ecdfdp{\HLModelFunc}^{-1}$, the resulting $\ecdfdp{\HLModelFunc}(\cdot)$ is $(L+1)\epsilon'$-DP independently of the number of needed evaluations of it during calls to Algorithm \ref{algo:invEcdf}.

  Next, for the evaluation of the statistics in lines \ref{ln:HL.expobs1}--\ref{ln:HL.expobs2} independent Laplacian random variables are used.  However, when we compare two adjacent datasets where only one instance differs, only $2$ of the $Q$ groups and the corresponding $8$ statistics are affected.  Hence, if all $4Q$ statistics in lines \ref{ln:HL.expobs1}--\ref{ln:HL.expobs2} are $\epsilon'$-DP, together they are $8\epsilon'$-DP.

  In summary, applying the classic composition rule for differential privacy we get that the algorithm is $(L+1)\epsilon' + 8\epsilon' = \epsilon$ -differentially private.
\end{proof}

\section{Additional experimental results}

Figures \ref{fig:ex.roc.bank.0.2} and \ref{fig:ex.roc.diabetic.0.05} show examples of ROC curves on the Bank and Diabetic datasets.

\begin{figure}[ht]
\vskip 0.2in
\begin{center}
\centerline{\includegraphics[width=\columnwidth]{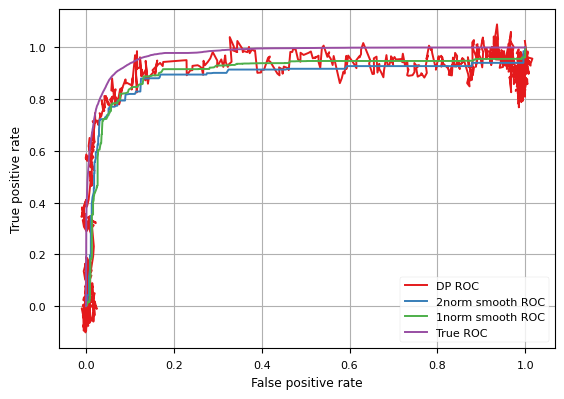}}
\caption{ROC curve for logistic regression on the Bank dataset, and $\epsilon$-DP curves with $\epsilon=0.2$.}
\label{fig:ex.roc.bank.0.2}
\end{center}
\vskip -0.2in
\end{figure}

\begin{figure}[ht]
\vskip 0.2in
\begin{center}
\centerline{\includegraphics[width=\columnwidth]{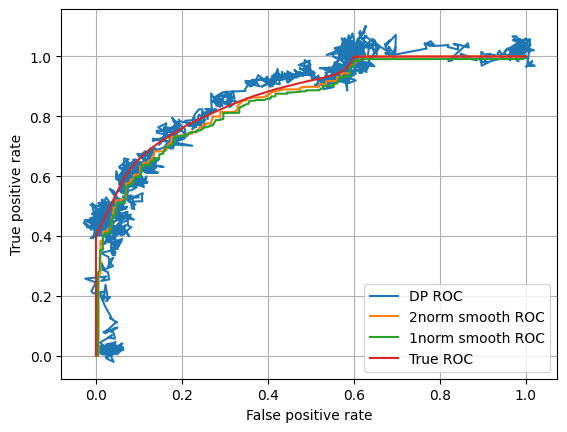}}
\caption{ROC curve for logistic regression on the Diabetic dataset, and $\epsilon$-DP curves with $\epsilon=0.05$.}
\label{fig:ex.roc.diabetic.0.05}
\end{center}
\vskip -0.2in
\end{figure}

Figure \ref{fig:HL.dia} show the relative MSE of the Hosmer-Lemeshow statistic on the Diabetes dataset as a function of $\epsilon$.

\begin{figure}[ht]
\vskip 0.2in
\begin{center}
\caption{Hosmer-Lemeshow statistic MSE for a logistic regression model on the Diabetes dataset}
\label{fig:HL.dia}
\end{center}
\vskip -0.2in
\end{figure}

\newpage
\appendix

\end{document}